\documentclass[10pt]{amsart}

\usepackage{array}
\usepackage{cancel}
\usepackage[british]{babel}
\usepackage{pifont,subfigure,graphicx}
\usepackage{caption}
\usepackage[colorlinks,citecolor=blue,urlcolor=blue]{hyperref}
\usepackage{url}
\usepackage{amsthm,amsmath,amssymb,amsfonts,mathrsfs,amsfonts,amstext,amsopn}
\usepackage{mathtools}
\usepackage{tikz}
\usepackage[normalem]{ulem}
\usepackage{hhline}
\usepackage[version=3]{mhchem}
\usepackage{texdraw}
\usepackage{extarrows}
\usepackage{bbm}
\usepackage{enumitem}
\usepackage[all]{xy}

\usepackage[tableposition=top]{caption}
\usepackage{longtable}
\usepackage{pgfplots}
\pgfplotsset{compat = newest}
\usetikzlibrary{decorations.fractals}
\usetikzlibrary{positioning}
\usetikzlibrary{arrows}

\numberwithin{equation}{section}
\theoremstyle{plain}
\newtheorem{theorem}{Theorem}[section]

\theoremstyle{definition}

\theoremstyle{remark}

\newcommand{\lt}{\left}
\newcommand{\rt}{\right}

\newcommand{\om}{\omega}
\newcommand{\N}{\mathbb{N}}

\newcommand{\R}{\mathbb{R}}
\newcommand{\Z}{\mathbb{Z}}
\newcommand{\cF}{\mathcal{F}}

\newcommand{\cT}{\Omega}

\newcommand{\cR}{\mathcal{R}}

\newcommand{\cS}{\mathcal{S}}

\allowdisplaybreaks
\makeatletter
\newcommand{\thickhline}{
    \noalign {\ifnum 0=`}\fi \hrule height 1pt
    \futurelet \reserved@a \@xhline
}
\newcolumntype{"}{@{\hskip\tabcolsep\vrule width 1pt\hskip\tabcolsep}}
\makeatother
\begin{document}
\title[Any SRN has a  Stationary Measure]
  {Any Stochastic Reaction Network has a  Stationary Measure}
\author{Carsten Wiuf $^{1}$}
\author{Chuang Xu $^{2}$}
\email{wiuf@math.ku.dk (Corresponding author)}
\address{$^1$
Department of Mathematical Sciences,
University of Copenhagen, 2100 Copenhagen, Denmark.}
\address{$^2$
Department of Mathematics\\
University of Hawai'i at M\={a}noa, Honolulu\\
96822, HI, US.}

\date{\today}

\noindent

\begin{abstract} 
In this note, we use a result by Harris (1957) to show that there always exists a stationary measure (not necessarily a distribution) on a closed  irreducible component of a stochastic reaction network. This measure might not be unique. In particular, any weakly reversible stochastic reaction network has a stationary measure on all closed irreducibe components, irrespective whether it is compelx balanced or not.
\end{abstract}


\maketitle

\section{Introduction}
Stochastic reaction networks (SRNs) are continuous-time Markov chains on  $\N^n_0$ modelling the stochastic dynamics of a reaction network, a collection of chemical reactions. In the past, these have been used to model many other natural processes that involve interactions between entities \cite{SS08,PCMV15,GMK17}.
 
A difficult problem seems to be to show the existence of a stationary distribution on an irreducible component of an SRN \cite{AK15,GMK17}.  A result in \cite{H57} makes it trivial to show the existence of a stationary measure. However, it leaves the problem of showing that the irreducible component is positive recurrent to infer the measure is a distribution. 

\section{Preliminaries}

\subsection{Markov Chains}

We define a class of CTMCs on $\N^n_0$ in terms of a finite set of jump vectors and  non-negative transition functions.  Let $\cT\subseteq\Z^n\!\setminus\!\{0\}$ be a finite set and $\cF=\lt\{\lambda_{\om}\colon \om\in\cT\rt\}$ a set of non-negative transition functions on $\N_0$,
$$\lambda_{\om}\colon\N_0^n\to\R_{\ge0},\quad \om\in\cT.$$
  The transition functions define a $Q$-matrix $Q=(q_{x,y})_{x,y\in \N_0^n}$ with  $q_{x,y}=\lambda_{y-x}(x)$, $x,y\in\N_0^n$, and subsequently, a class of CTMCs $(Y_t)_{t\ge 0}$ on $\N_0^n$ by assigning an initial state $Y_0\in\N_0^n$.
For convenience, we identify the class of CTMCs with $(\Omega,\cF)$.

A subset $C\subseteq\N^n_0$ is an \emph{irreducible component} (aka communicating class) if  there is positive probability of jumping from $x$ to $y$ for any $x,y\in C$ in a finite number of steps, that is, there   exists a sequence of states $x_0,\ldots,x_m$, such that $x=x_0$, $y=x_m$ and $\lambda_{\omega_i}(x_i)>0$ with $\omega_i=x_{i+1}-x_i\in\cT$,  $i=0,\ldots,m-1$, for some $m\in\mathbb{N}_0$. Furthermore, $C$ should be maximal in that sense. An irreducible component is \emph{closed} if for $x\in C$ and $\lambda_\omega(x)>0$ for some $\omega\in\Omega$, then $x+\omega\in C$.

 A \emph{non-zero} measure $\pi$ on a closed irreducible component $C\subseteq \N_0^n$ of $(\Omega,\cF)$ is a  \emph{stationary measure} of $(\Omega,\cF)$ if $\pi$   is invariant for the $Q$-matrix, that is, if $\pi$ is a  non-negative equilibrium of the  {\em master equation} \cite{G92}:
 \begin{equation}\label{eq:master}
 0=\sum_{\omega\in\cT}\lambda_{\omega}(x-\omega)\pi(x-\omega)-\sum_{\omega\in\cT}\lambda_{\omega}(x)\pi(x),\quad 
 x\in C,
 \end{equation}
where for convenience, we define $\pi(x)=0$ if $x\not\in C$.

\subsection{SRNs}
 
 A reaction network is a  finite collection $\cR$ of  \emph{reactions} $y\to y'$, where the source and the target of a reaction are non-negative linear combinations of  \emph{species} $\cS$. The source and target nodes are called \emph{complexes}.  
 
One might specify a continuous-time Markov chain $(X_t)_{ t\geq 0}$ on the ambient space $\N_0^n$ of a reaction network, where $n=\#\cS$ is the cardinality of $\cS$ and $X_t$ is the vector of species counts at time $t\ge 0$.  The complexes are represented as elements of $\N_0^n$ via the natural embedding, assuming $\cS=\{S_1,\ldots,S_n\}$ is ordered.  If  $\cR=\{y_1\to y'_1,\ldots,y_r\to y'_r\}$ and reaction $y_k\to y'_k$ occurs at time $t$, then the new state is $X_t=X_{t-}+\xi_k$, where $X_{t-}$ denotes the previous state and $\xi_k=y'_k-y_k$.
The stochastic process can be given as
\begin{align*}
X_t=X_0+\sum_{y_k\to y'_k\in \cR}\xi_k Y_k\left(\int_0^t \eta_k(X_s)\,ds\right),
\end{align*}
where $Y_k$ are independent  unit-rate Poisson processes and $\eta_k\colon\mathbb{N}_0^n\to[0,\infty)$ are intensity functions \cite{AK15,Ke11,markov}. By varying the initial vector of species counts $X_0$, a whole family of Markov chains is associated with the SRN. An SRN is denoted $(\cR,\eta)$, where $\eta=(\eta_1,\ldots,\eta_r)$.

 Several reactions might give rise to the same jump vector, thus in the teminology of the previous section, $\cT=\{y'_k-y_k|k=1,\ldots,r\}$, and
$$\lambda_\omega(x)=\sum_{y_k\to y_k'\in\cR\colon y'_k-y_k=\omega} \eta_k(x).$$

\section{Existence of stationary measure}

\begin{theorem}
Let $C\subseteq\N_0^n$ be a closed irreducible component of $(\cT,\cF)$. Then, there exists a stationary measure on $C$.
\end{theorem}

\begin{proof}
If $C$ is finite, then it follows trivially from Markov chain theory that there exists a stationary distribution on $C$, hence also a stationary measure. If $C$ is countable infinite, then it is isomorphic to $\N_0$ as sets. If the chain is recurrent, the existence of a stationary measure follows from \cite[Theorem 3.5.2]{markov}. If the chain is transient, then the conclusion follows from \cite[Corollary]{H57} and \cite[Theorem 3.5.1]{markov}, noting that the set of states accessible to any given state $x\in C$  is finite, in fact $\le\#\Omega$.
\end{proof}

The existence is well known if $C$ is recurrent: if it is positive recurrent then there exists a unique stationary disstibution, and if it is null recurrent, then there exists a unique  stationary measure, up to a scaling factor \cite{markov}. {In the transient case, the measure might not be unique.}  If $C$ is transient and non-explosive, then there cannot be a stationary distribution, only a measure  \cite[Theorem 3.5.3]{markov}. If $C$ is transient and  explosive, then there might be a stationary distribution.

\section*{Acknowledgements}

CW acknowledges support from the Novo Nordisk Foundation (Denmark), grant NNF19OC0058354. CX acknowledges the support from TUM University Foundation and the Alexander von Humboldt Foundation and an internal start-up funding from the University of Hawai'i at M\={a}noa.


\end{document}